\documentclass[sigconf]{acmart}

\usepackage{amsmath,amssymb,amsfonts}
\usepackage{algorithmic}
\usepackage{graphicx}
\usepackage{grffile}
\usepackage{subcaption}
\usepackage{textcomp}
\usepackage{subcaption}

\usepackage{xcolor}
\usepackage{pifont} 
\usepackage{multirow}
\usepackage{rotating}
\usepackage[table]{xcolor}
\newcommand{\cmark}{\checkmark}
\newcommand{\xmark}{\ding{55}}
\usepackage{threeparttable}
\usepackage{subcaption}
\def\BibTeX{{\rm B\kern-.05em{\sc i\kern-.025em b}\kern-.08em
    T\kern-.1667em\lower.7ex\hbox{E}\kern-.125emX}}

\usepackage{pifont} 

\usepackage[ruled,vlined,linesnumbered]{algorithm2e}
\usepackage{wrapfig}

\let\oldnl\nl
\newcommand{\nonl}{\renewcommand{\nl}{\let\nl\oldnl}}

\usepackage[font=small,skip=3pt]{caption}
\usepackage{subcaption}
\usepackage{booktabs} 
\usepackage{bbm}
\usepackage{amsthm}

\newtheorem{corollary}{Corollary}

\SetKwInput{KwData}{Input}
\SetKwInput{KwResult}{Output}

\usepackage{mathtools}
\newtheorem{lemma}{Lemma}

\usepackage{enumitem}
\usepackage{setspace}
\DontPrintSemicolon

\usepackage{makecell}

\AtBeginDocument{%
  \providecommand\BibTeX{{%
    Bib\TeX}}}

\acmConference[ICS 2026]{Make sure to enter the correct
  conference title from your rights confirmation email}{July 06--09,
  2026}{Belfast, Northern Ireland}

\author{S M Shovan}
\authornote{These authors contributed equally to this work.}
\email{sskg8@mst.edu}
\affiliation{
  \institution{Missouri University of Science and Technology}
  \city{Rolla}
  \state{MO}
  \country{USA}
}

\author{Arindam Khanda}
\authornotemark[1]
\email{akkcm@mst.edu}
\affiliation{
  \institution{Missouri University of Science and Technology}
  \city{Rolla}
  \state{MO}
  \country{USA}
}

\author{S M Ferdous}
\email{sm.ferdous@pnnl.gov}
\affiliation{
  \institution{Pacific Northwest National Laboratory}
  \city{}
  \state{}
  \country{USA}
}

\author{Sajal K. Das}
\email{sdas@mst.edu}
\affiliation{
  \institution{Missouri University of Science and Technology}
  \city{Rolla}
  \state{MO}
  \country{USA}
}

\author{Mahantesh Halappanavar}
\email{hala@pnnl.gov}
\affiliation{
  \institution{Pacific Northwest National Laboratory}
  \city{}
  \state{}
  \country{USA}
}

\newcommand{\stlp}{{\sc StLP}}
\newcommand{\itlp}{{\sc ItLP}}
\newcommand{\dblp}{{\sc DynLP}}
\newcommand{\LZero}{{\sc L_0}}
\newcommand{\LOne}{{\sc L_1}}

\begin{document}

\title{DynLP: Parallel Dynamic Batch Update for Label Propagation in Semi-Supervised Learning \\ \large{(to be published in the ACM International Conference on Supercomputing (ICS 2026))}}

\begin{abstract}
 Semi-supervised learning aims to infer class labels using only a small fraction of labeled data. In graph-based semi-supervised learning, this is typically achieved through label propagation to predict labels of unlabeled nodes. 
  However, in real-world applications, data often arrive incrementally in batches. Each time a new batch appears, reapplying the traditional label propagation algorithm to recompute all labels is redundant, computationally intensive, and inefficient. To address the absence of an efficient label propagation update method, we propose \dblp{}, a novel GPU-centric Dynamic Batched Parallel Label Propagation algorithm that performs only the necessary updates, propagating changes to the relevant subgraph without requiring full recalculation. By exploiting GPU architectural optimizations, our algorithm achieves on average $13 \times$ and upto $102 \times$ speedup on large-scale datasets compared to state-of-the-art approaches.
\end{abstract}





\maketitle
\begin{center}

\end{center}

\section{Introduction}

In many machine learning applications, data evolve over time, and there is a need to analyze such dynamic data efficiently without recomputing results from scratch. These tasks become even more challenging when the training data are limited or expensive to obtain, as labeling often requires human input, sometimes from domain experts. Many learning tasks, including sentiment analysis on text, web page categorization, speech analysis, and medical image classification, fall into this dynamic and sparsely labeled data category. In such scenarios, conventional supervised learning is inadequate as it relies on a large amount of labeled data.

Semi-supervised learning (SSL)~\cite{chapelle2006semi} models address this imbalance between labeled and unlabeled data. Among SSL methods, graph-based approaches have gained popularity due to their accuracy and computational efficiency \cite{subramanya2014graph,CHONG2020216,song2022graph}. However, most existing studies focus on static data, overlooking the inherent dynamism in many learning settings. In this paper, we study graph-based semi-supervised learning (GSSL) for dynamic data. Specifically, we design and implement efficient parallel algorithms for label propagation, a widely used inference approach for GSSL, on evolving graphs.

A typical GSSL framework consists of two key phases~\cite{zhu2003semi}: (1) graph construction from data and (2) label inference on the constructed graph. For datasets that are not naturally represented as graphs (e.g., collections of images), a similarity graph is first constructed—most commonly using neighborhood-based methods such as $k$-nearest neighbors (kNN). The label inference phase then predicts the labels of unlabeled nodes using the few labeled ones (seed nodes). Label propagation (LP) and its variants are among the earliest~\cite{zhu2003semi,zhou2004learning} and most widely used methods~\cite{song2022graph} for this task. LP solves a quadratic optimization problem involving the graph Laplacian. Although an analytical solution exists, it is computationally expensive due to the dense matrix operations involved. A more scalable alternative adopts an iterative random-walk-like process, which we call \itlp{}: starting from known labels for seed nodes and arbitrary initial labels for unlabeled ones, it iteratively updates each node’s label as the average of its neighbors’ labels while keeping the labeled nodes fixed. This iterative approach is guaranteed to converge to the analytic solution. Beyond GSSL, LP methods are also used in applications such as community detection~\cite{raghavan2007near} and in augmenting graph neural networks (GNNs)~\cite{wang2021combining}.

In this paper, we adopt a dynamic algorithm model, where batches of data changes (i.e., graph node addition and deletion) arrive, and the goal is to maintain accurate labels for all vertices as the graph evolves. Since each batch of nodes is inserted/deleted at discrete time steps, parallel computing can be leveraged to accelerate updates. A naive approach would be to simply augment the neighborhood graph (e.g., via kNN or $\epsilon$-neighborhood construction) and rerun the entire LP procedure from scratch. Wagner et al.~\cite{wagner2018semi} proposed a streaming algorithm for temporal label propagation (\stlp{}), which can be adapted to our dynamic setting. Their approach uses a graph reduction technique inspired by the short-circuit operator in electrical networks to improve memory efficiency. However, when adapted to our setting, we show that \stlp{} still requires dense matrix multiplications after every batch updates, making it unsuitable for large-scale datasets.

Our proposed Dynamic Batch Parallel Label Propagation (\dblp{}) algorithm improves upon \itlp{} and \stlp{} by maintaining connected components across batches of nodes and performing iterative label propagation on a reduced graph representation. We compare \dblp{} with our parallel implementations of \itlp{} and\stlp{} on diverse synthetic and real-world datasets and the scalability grows proportional to the dataset size and the number of batches.

Below, we summarize our contributions:
\begin{itemize}
\vspace{-0.05in}
\item We provide \stlp{}, a GPU parallel implementation of the streaming label propagation algorithm of~\cite{wagner2018semi}.
\item We develop \dblp{}, a novel, dynamic parallel label propagation algorithm that supports both insertions and deletions of data points. It addresses the inefficiencies of \stlp{} using connected component-based efficient update while avoiding redundant operations. 
\item We implement \dblp{} on GPUs and compare it against \stlp{} and \itlp{} in terms of speedup, convergence rate and accuracy in both real and synthetic sparse datasets.
\item Our results show that \dblp{} achieves, on average, a 13$\times$ speedup (up to 102$\times$) over state-of-the-art iterative methods, while being up to 100$\times$ more memory efficient than harmonic-solution-based approaches for sparse graphs.
\end{itemize}

\vspace{-0.1in}
\section{Related work}
We review label propagation techniques for graph-based semi-supervised learning on static and dynamic cases. 

\vspace{-0.05in}
\subsection{Static} 
In the static case, labels are inferred only for the unlabeled nodes already present in the graph~\cite{song2022graph}.

Zhu et al.~\cite{zhu2003semi} model the label function over the graph as a Gaussian random field, where nearby nodes are encouraged to have similar label values. They show that the mean of this field minimizes a quadratic energy function defined by the graph Laplacian (detailed in \S~\ref{sec:prelim}). The resulting harmonic function has a closed-form solution involving the inverse of a submatrix of the Laplacian, which can alternatively be computed through iterative label propagation methods, which is the focus of our paper.
Zhu et al. first formulated the problem with a Gaussian fields-based objective function \cite{zhu2003semi}. In a subsequent study \cite{zhu2003semib}, they showed that this objective can be written in combinatorial Laplacian form, enabling optimal minimization. However, the resulting solution required $O(n^2)$ space, motivating later work on more efficient alternatives. 
The direct analytical method can be accelerated using low-rank matrix approximations of the graph \cite{delalleau2005efficient}, while other works employ quantization based on cluster centroids \cite{valko2012online}. Later, researchers such as \cite{li2016graph} further improved label-propagation performance by enhancing the quality of the underlying graph structure. Sparse-graph construction using principal component analysis has also been explored in studies such as \cite{hua2022robust}.
More recently, researchers have incorporated random-walk-based label-propagation techniques, in which labeled nodes are treated as absorbing states under the assumption that no alternative paths are explored once a labeled node is reached \cite{desmond2021semi}. This family of approaches has been extended using guided random walks \cite{fazaeli2022guidedwalk}, allowing the exploration of additional paths associated with the same class. Random-walk-based label propagation has also been applied to graph-clustering tasks \cite{song2023instance}. Although rank-reduction strategies and absorbing random-walk formulations provide modest improvements in scalability, they often sacrifice guarantees on solution quality. Moreover, none of these approaches is designed for sparse graphs, nor do they parallelize computation to exploit modern high-performance architectures.

\vspace{-0.05in}
\subsection{Dynamic}  
In real-world applications, unlabeled nodes often arrive incrementally over time, either as a continuous stream or in multiple batches, which constitutes an inductive scenario. Recent studies have explored various forms of dynamicity, such as the arrival of new samples with distinct feature spaces~\cite{wu2023online, din2020online} or the emergence of entirely new class labels~\cite{zhu2020semi}.

Zhu et al.~\cite{zhu2009somed} proposed an online solution with a no-regret formulation, where the cumulative error propagated from batch to batch remains close to zero. Sindhwani et al.~\cite{sindhwani2005beyond} addressed a related challenge in which the unlabeled nodes need not follow the same distribution as the labeled ones. Huang et al.~\cite{huang2015online} applied label propagation for dataset annotation, where label prediction was restricted to the current batch of unlabeled nodes only.

However, most of these methods suffer from scalability issues due to their $O(n^2)$ space complexity. Many methods assume that incoming vertices have known similarities to all existing vertices, which is unrealistic when many pairwise similarities are unknown. In such sparse settings, the inherent $O(n^2)$ space complexity restricts scalability~\cite{duff1988sparsity}.
Later, Ravi et al.~\cite{ravi2016large} 
proposed an approximate solution, which does not leverage current predictions to infer future labels and scales linearly with the number of vertices. 

To mitigate the memory bottleneck, Wagner et al.~\cite{wagner2018semi} introduced a short-circuiting approach that represents each ground-truth class using only two representative nodes without information loss. Nevertheless, for $n$ unlabeled vertices, the memory requirement remains $O(n^2)$, which is high, especially when labeled nodes are far fewer than unlabeled ones.

Recent approaches employ graph neural networks (GNNs) to handle unreliable or evolving nodes by leveraging implicit semantic relationships~\cite{yu2024gnn}. Prototype-based learning methods summarize streams of unlabeled nodes into representative prototypes with varying levels of granularity~\cite{din2024reliable}. Other approaches use generative feature models~\cite{li2019incremental}, rather than simple similarity scores, to propagate labels with the goal of capturing latent feature representations. Although these machine-learning-based methods introduce promising ideas, their training and inference times often struggle to scale in highly dynamic environments where changes occur frequently. Furthermore, many of these models require a substantial amount of ground-truth labels to achieve strong predictive performance, making them unsuitable for scenarios where only a small fraction of labeled data is available.


The study by Bunger et al.~\cite{bunger2022empirical} empirically compared various time series distance measures under the assumption that the underlying graph is fully connected, overlooking scenarios in which similarities between certain vertex pairs may be unknown or undefined. This assumption limits scalability and disregards the feasibility of performing label propagation on sparse graphs~\cite{van2020survey}.
Due to the absence of scalable label propagation algorithms, Song et al.~\cite{song2022graph} emphasized the need for parallelization in future research, noting that recent approaches should achieve time complexity that scales linearly with the number of samples. Only a few studies, such as Covell et al.~\cite{covell2013efficient}, have attempted to parallelize label propagation. However, their approach targets static graphs in which the labels of both known and unknown samples are already available. When applied to multi-batch settings, such methods typically recompute the entire process for each batch of changes, resulting in redundant computation. 
The lack of scalable label propagation approaches on large-scale sparse dynamic graphs motivates us to design \dblp{}. 


\color{black}
\vspace{-0.05in}
\section{Preliminaries}
\label{sec:prelim}
We model data as a weighted, undirected graph
$G = (V, E)$, where $V$ is the vertex set denoting the data points and $E \subseteq V \times V$ is the edge set denoting the similarity among the vertices. 
Let $\mathcal{N}(u) = \{ v \in V : w_{u,v} > 0 \}$ denote the neighbor set of $u$, where $w : V \times V \to \mathbb{R} \ge 0$ is a similarity function on vertex pairs that assigns a similarity value as the edge weight. The Laplacian of the graph $G$ is defined as $\mathbf{L}=\mathbf{D}-\mathbf{W}$, where $\mathbf{D}$ and $\mathbf{W}$ are the degree and weight matrices, respectively.

Given a graph $G$ with a small labeled subset $V^L \subseteq V$ and unlabeled nodes $V^U = V \setminus V^L$, and a similarity function $w$ on vertex pairs, semi-supervised learning infers labels for all nodes in $V = V^L \cup V^U$. Here, we consider a binary classification setup with labels $1$ and $0$. 


\subsection{Label Propagation} 
The simplest yet effective way to infer labels for unlabeled nodes is \emph{label propagation}~\cite{zhu2003semi}. It enforces smoothness iteratively: adjacent vertices with large edge weight should have similar label scores, while labels on $V^L$ remain fixed. At each iteration, every unlabeled node $u$ replaces its label $F_u$ with the average of its neighbors’ labels (unweighted graph) or the weighted average (weighted graph). 
Except for $u \in V^L$ that have the ground-truth labels $Y_u$, all the nodes that appear from time $1$ to $t$ are iteratively updated until convergence using the following equation.
\begin{equation}
\label{eq:label_prop}
F_u^{(k+1)}=
\begin{cases}
\frac{1}{d(u)}\sum_{v\in V} w(u,v)F_{v}^{(k)},\quad \forall u\in V \setminus V^L,\\
Y_u, \quad \forall u\in V^L,
\end{cases}
\end{equation}
where $d(u)=\sum_{v\in V} w(u,v)$ and $k$ denotes the iteration identifier.
In graph-based label propagation for binary classification, the harmonic solution also can be defined as a function $F: V \to \mathbb{R}^2$ that matches the given labels $Y_u$ on $u \in V^L$, and minimizes the energy function $\frac{1}{2} \sum_{(u, v) \in E} w_{u,v}\big(F_u - F_{v}\big)^2$~\cite{zhu2003semi}. The closed form formula for the unknown part of the solution can be derived as:
\begin{equation}
    F^U = -\mathbf{L}_{UU}^{-1}\mathbf{L}_{UL}F^L,
\end{equation}
where $\mathbf{L} =\begin{pmatrix}
\mathbf{L}_{LL} & \mathbf{L}_{LU}\\
\mathbf{L}_{UL} & \mathbf{L}_{UU}
\end{pmatrix}$ is the graph Laplacian after indexing vertices as labeled $L$ first, then unlabeled $U$.


\subsection{Dynamic Graphs and Temporal Label Propagation} 
A dynamic graph $G_t = (V_t, E_t)$ models data that evolves over time. At discrete time $t$, new data may appear as vertices $\Delta^{Ins}_t$, such that $V_{t+1} = V_t \cup \Delta^{Ins}_t$. On the other hand, some data  can become irrelevant and can be considered as deleted vertices $\Delta^{Del}_t$. Together, we denote all the changed vertices as $\Delta_t = \{\Delta^{Ins}_t, \Delta^{Del}_t\}$. Typically, $\Delta_t$ contains few or no labeled vertices. 
Therefore, semi-supervised learning aims to infer vertex labels at time step $t+1$ using the labeled vertices $V^L \subseteq V_{t+1}$, the unlabeled vertices $V^U = V_{t+1} \setminus V^L$, and the similarity function $w$. Note that labels of vertices with ground truth remain constant over time, whereas labels of vertices without ground truth may change due to the influence of newly appeared/disappeared vertices. 
Restricting $V^L$ to vertices with ground truth only, the task of predicting labels for the unlabeled vertices at time $t+1$ reduces to finding the harmonic solution on the updated graph $G_{t+1}$. Algorithm~\ref{alg:static} follows this approach. It first deletes vertices $u \in \Delta^{Del}_t$ and related edges from the existing graph. Then it adds vertices $u \in \Delta^{Ins}_t$ to $V_{t+1}$ and edges ${(u,v) : v \in V_{t+1}, u \in \Delta^{Ins}_t}$ to $E_{t+1}$. Finally the algorithm recomputes the harmonic solution $F^U$ for all unlabeled vertices $V^U = V_{t+1} \setminus V^L$, where $V^L$ is the set of vertices with ground truth.

\begin{algorithm}[h]
\caption{Label Recomputation}
\label{alg:static}
\DontPrintSemicolon
\small
\KwIn{Similarity graph $G_t = (V_t, E_t)$, label vector $F^L$ for vertices with ground truth, batch of new vertices $\Delta_t = \{\Delta^{Ins}_t, \Delta^{Del}_t\}$
}
\KwOut{Updated label vector $F^U$ for unlabeled vertices.}
\tcc{Step 1: Update Graph}
Initialize $V_{t+1} \gets V_t, \quad E_{t+1} \gets E_t$\;
\For{$u \in \Delta^{Del}_t$}{
    \For{$v \in V_{t+1}$}{
            Delete edge $(u,v)$ from $E_{t+1}$\;
    }
    $V_{t+1} \gets V_{t+1} \setminus u$\;  
}
\For{$u \in \Delta^{Ins}_t$}{
    \For{$v \in V_{t+1}$}{
            Add edge $(u,v)$ with weight $w_{v,u}$ into $E_{t+1}$\;
    }
    $V_{t+1} \gets V_{t+1} \cup u$\;  
}
\tcc{Step 2: Update labels}
$\mathbf{L} = \begin{bmatrix} \mathbf{L}_{LL} & \mathbf{L}_{LU} \\ \mathbf{L}_{UL} & \mathbf{L}_{UU} \end{bmatrix} \gets \mathbf{D}_{t+1} - \mathbf{W}_{t+1}$ \tcp{Compute Laplacian}
$F^U \gets -\mathbf{L}_{UU}^{-1}\mathbf{L}_{UL}F^L$ \tcp{Compute Harmonic Solution}
\end{algorithm}

Although this approach is straightforward, it becomes computationally inefficient on large incremental graphs.
Restarting propagation after each new batch forces all previously processed nodes to participate in every iteration and allows any previously labeled node (without ground truth) to change its label due to the new nodes, leading to rapidly growing computation time and memory.


To avoid this issue, the \textit{short circuiting} method is proposed for streaming incremental graphs~\cite{wagner2018semi}. It reduces dimensionality by contracting all vertices with the same ground truth label into one representative node per class. For each representative, its edges to other vertices are replaced by the parallel edge sum. This compact graph preserves the information and enables memory-efficient temporal label propagation.

Although effective on dense graphs, this method relies on the Laplacian and its inverse. Since the inverse of a sparse matrix is typically dense~\cite{ponte2024computing}, the approach cannot exploit sparse linear algebra and is neither scalable nor well-suited to large real-world graphs, which are usually sparse.
Second, for a set of changes $\Delta_t$, the method recomputes labels for all vertices in $V^U$ from scratch.
However, in practice, influence decays with propagation, so changed vertices in $\Delta_t$ affect only nodes within a limited number of hops.
These motivate us to design a parallel label update approach that avoids redundant recomputation of the harmonic solution.

\begin{algorithm}[!h]
\caption{\dblp{} Update}
\label{alg:proposed}
\DontPrintSemicolon
\small
\KwIn{Similarity graph $G_t = (V_t, E_t)$, label vector $F^L = \{F^{\LZero{}}, F^{\LOne{}}\}$ for vertices with ground truth, label vector $F^U_t$ for the vertices in $V_t$ without ground truth,  batch of new vertices $\Delta_t = \{\Delta^{Ins}_t, \Delta^{Del}_t\}$.
}
\KwOut{Updated label vector $F^U_{t+1}$ for all vertices in $V_{t+1}$ without ground truth.}

\tcc{Step 1: Change Adjustment and Sparsification}
Initialize $V_{aff} \gets \emptyset$\label{code:del_aff_start}\;
Initialize $V_{t+1} \gets V_t, \quad E_{t+1} \gets E_t$\;

\For{$u \in \Delta^{Del}_t$}{
    $V_{aff} \gets V_{aff} \cup \mathcal{N}(u)$\;
    Remove $u$ from $V_{t+1}$ and related edges from $E_{t+1}$\;
} \label{code:del_aff_end}
Initialize an empty graph $G'(V', E')$\label{code:graph_cont_start}\;
$V' \gets \Delta^{Ins}_t$\label{code:graph_cont_mid1}\;
\For{$u \in \Delta^{Ins}_t$}{
    
    \For{$v \in V_{t+1}$}{
        \If{$Sim(v,u) > \tau$}{
            Add edge $(u,v)$ with weight $Sim(u,v)$ into $E_{t+1}$\;
            \If{$v \in \Delta^{Ins}_t$}{
                Add edge $(u,v)$ with weight $Sim(u,v)$ into $E'$\;
            }
        }
    }
    $V_{aff} \gets V_{aff} \cup \{u\} \cup \mathcal{N}(u)$\;
}
$\mathcal{C} \gets $ \texttt{FindConnectedComponents}$(G'(V', E'))$\;

\tcc{Step 2: Label Initialization}
Let $\LZero{}$ is the super node consisting of the vertices $u \in V^{\LZero{}}$ with label $0$ only.\;
Let $\LOne{}$ is the super node consisting of the vertices $u \in V^{\LOne{}}$ with label $1$ only.\;
\For{$c_i \in \mathcal{C}$}{
    $\mathcal{W}_{c_i}^{\LZero{}} \gets \sum_{u \in c_i}\sum_{v \in \LZero{}} w(u,v)$\label{code:weight1}\;
    $\mathcal{W}_{c_i}^{\LOne{}}(c_i) \gets \sum_{u \in c_i}\sum_{v \in \LOne{}} w(u,v)$\;
    \For{$u \in c_i$}{
        $F_u \gets 0.5 - \frac{\mathcal{W}_{c_i}^{\LZero{}}}{2 \cdot (\mathcal{W}_{c_i}^{\LZero{}} + \mathcal{W}_{c_i}^{\LOne{}})} + \frac{\mathcal{W}_{c_i}^{\LOne{}}}{2 \cdot (\mathcal{W}_{c_i}^{\LZero{}} + \mathcal{W}_{c_i}^{\LOne{}})}$\;
    }
}
\tcc{Step 3: Iterative Propagation}
\While{$V_{aff} \neq \emptyset$}{
    \For{$u \in V_{aff}$ in parallel}{
        $\mathcal{W}^{all} \gets \sum_{v \in \mathcal{N}(u)}w(u,v)$\;
        $\mathcal{W}_u^{\LZero{}} \gets \sum_{v \in \LZero{}}w(u,v)$\;
        $\mathcal{W}_u^{\LOne{}} \gets \sum_{v \in \LOne{}}w(u,v)$\;
        
        $F'_u \gets F_u+(0-F_u)  \frac{\mathcal{W}_u^{\LZero{}}}{\mathcal{W}^{all}}+(1-F_u)  \frac{\mathcal{W}_u^{\LOne{}}}{\mathcal{W}^{all}} +\sum_{v \in \mathcal{N}(u) \setminus \{{V^{\LZero{}}, V^{\LOne{}}\}}}(F_v-F_u) \frac{w(u v)}{\mathcal{W}^{all}}$\;
        \If{$F'_u - F_u > \delta$}{
            $V_{aff} \gets V_{aff} \cup \mathcal{N}(u)$\;
            $F'_u \gets F_u$\;
        }
        \Else{Remove $u$ from $V_{aff}$\;}
    }
}
\end{algorithm}

\section{Proposed \dblp{}}

Here, we design a parallel label update algorithm, \dblp{} (Algorithm~\ref{alg:proposed}), to predict labels of vertices in dynamic graphs while avoiding full label recomputation. Our algorithm considers a semi-supervised learning scenario for binary classification with a very few labeled ground-truth set denoted as $V^L = \{V^{\LOne{}} \cup V^{\LZero{}}\}$, where $V^{\LOne{}}$ and $V^{\LZero{}}$ are the sets of vertices of classes $1$ and $0$, respectively and $V^{\LOne{}} \cap V^{\LZero{}} = \emptyset$.
When a new batch of data $(\Delta_t = \{\Delta^{Ins}_t, \Delta^{Del}_t\})$ arrives, 
\dblp{} computes fractional labels in $[0,1]$ for the unlabeled vertices $V^U$ to facilitate a partition into class $0$ and class $1$ sets. Algorithm~\ref{alg:proposed} consists of three steps: (i) Change Adjustment and Sparsification, (ii) Label Initialization, and (iii) Iterative Propagation.

\textbf{Change Adjustment and Sparsification: } 
In a similarity graph $G_t$, the label of a vertex from an unknown class is often influenced by the aggregation of the labels of its neighbors. Therefore, deleting a vertex in a similarity graph can impact the labels of its neighbors. Accordingly, \dblp{} marks the deletion affected vertices by visiting the neighbors of each vertex $u \in \Delta^{Del}_t$ in parallel and storing them in $V_{aff}$, a list of affected vertices that require further processing (Algorithm~\ref{alg:proposed} Line~\ref{code:del_aff_start}-\ref{code:del_aff_end}). Similarly, for inserted vertices, both the vertex $u \in \Delta^{Ins}_t$ and its neighbors $\mathcal{N}(u)$ are marked as affected and require label updates. However, note that newly arrived vertices typically have unknown labels and may require multiple iterations of label propagation to reach a stable label, whereas existing neighboring vertices already have assigned labels from the previous time stamp and should require only a few iterations of label propagation to adjust their labels. 

To improve the efficiency of label update we follow two approaches: (1) A compact graph representation (sparsification) to reduce the size of computation, (2) A good initial labeling for the new unlabeled vertices in $\Delta^{Ins}_t$ such that the label propagation is expected to converge in fewer iterations.
We observe that the data points with common features often show high similarity among themselves, leading to higher edge weights compared to the edge weight connecting two dissimilar data points~\cite{ranjan2025securing}. 
It enables us to design a method to predict the labels of similar incoming data points early through graph sparsification.

Given the incoming vertices $\Delta_t$, Algorithm~\ref{alg:proposed}, Lines~\ref{code:graph_cont_start}-\ref{code:graph_cont_mid1}  first constructs a graph $G' = (V', E')$ solely with the newly arriving vertices such that $V' = \Delta^{Ins}_t$. An edge between a pair $(u, v)$, $\forall u,v \in \Delta^{Ins}_t$, is included in $E'$ only if the similarity between the vertices $Sim(v,u)$ exceeds a predefined threshold $\tau$. This approach yields disjoint subgraphs in which the vertices of each connected subgraph share strong commonality. Consequently, vertices within a connected subgraph are likely to receive similar fractional labels through label propagation. Throughout the experiments, we set the value of $\tau$ to the average of the edge weights in each dataset.

To improve the scalability of label propagation and reduce the number of iterations required for label convergence, the sparsification step identifies the connected components (Let $\mathcal{C}$) of $G'$ and treats each component $c_i \in \mathcal{C}$ as a supernode.


\textbf{Label Initialization: } 
As the vertices in a supernode are considered similar, a supernode can be initialized with a single value representing the initial label for all its vertices. In the absence of additional information, each supernode can be initialized with a label value of $0.5$, representing the midpoint between class labels $0$ and $1$. 
However, starting label propagation with vertex labels closer to their true values takes fewer iterations to converge. Therefore, assuming data points with close label values have higher mutual similarity weights, we leverage similarities between the supernodes and two initial vertex sets with ground truth $V^{\LZero{}}$ and $V^{\LOne{}}$ to predict better initial labels for the supernodes $c_i \in \mathcal{C}$. 

Treating $V^{\LZero{}}$ as a special supernode $\LZero{}$, the similarity weight between any supernode $c_i \in \mathcal{C}$ and $\LZero{}$ can be computed as the edge weight sum $\mathcal{W}_{c_i}^{\LZero{}} = \sum_{u \in c_i}\sum_{v \in \LZero{}} w(u,v)$ (Algorithm~\ref{alg:proposed}, Line~\ref{code:weight1}). 
Similarly, treating $V^{\LOne{}}$ as a supernode $\LOne{}$, the weight between $c_i$ and $\LOne{}$ is $\mathcal{W}_{c_i}^{\LOne{}} = \sum_{u \in c_i}\sum_{v \in \LOne{}} w(u,v)$.

Leveraging the edge weights between the supernodes in $\mathcal{C}$ and the supernodes with ground truth, each vertex $u$ in a supernode $c_i \in \mathcal{C}$ can be initialized with
\[
F_u = 0.5 + (0 - 0.5)\frac{\mathcal{W}_{c_i}^{\LZero{}}}{\bigl(\mathcal{W}_{c_i}^{\LZero{}} + \mathcal{W}_{c_i}^{\LOne{}}\bigr)} + (1-0.5)\frac{\mathcal{W}_{c_i}^{\LOne{}}}{\bigl(\mathcal{W}_{c_i}^{\LZero{}} + \mathcal{W}_{c_i}^{\LOne{}}\bigr)},
\]
where the first term, $0.5$, provides a neutral initialization, the second term reflects the similarity contribution of $\LZero{}$ and reduces the value toward $0$, and the third term reflects the contribution of $\LOne{}$ and increases the value toward $1$.

\textbf{Iterative Propagation: } This step considers the updated graph $G_{t+1} = (V_{t+1}, E_{t+1})$, where $V_{t+1} = V_t \setminus \{\Delta^{Del}_t \cup \Delta^{Ins}_t$\} and $E_{t+1}$ includes all edges with the vertices $V_{t+1}$.
As updating a vertex’s label can affect its neighbors, the iteration begins by updating the labels of the vertices $u \in \Delta_t$ along with their neighbors $\mathcal{N}(u)$. There are three kinds of vertices that can impact the label of a vertex $u$:
\begin{enumerate}
    \item Vertices with ground truth class $0$, denoted as a supernode $\LZero{}$. The similarity weight between $u$ and $\LZero{}$ is the edge-weight sum $\mathcal{W}_u^{\LZero{}} = \sum_{v \in \LZero{}} w(u,v)$.
    \item Vertices with ground truth class $1$. They have similarity weight $\mathcal{W}_u^{\LOne{}} = \sum_{v \in \LOne{}} w(u,v)$ with vertex $u$.
    \item Neighbor vertices $v \in \mathcal{N}(u) \setminus V^{\LZero{}} \setminus V^{\LOne{}}$ from timestamp $t+1$ or earlier, without any ground truth.
\end{enumerate}
Therefore, the label of $u$ at each iteration can be updated as:
\begin{equation}
\begin{split}
F'_u &= F_u+(0-F_u)  \frac{\mathcal{W}_u^{\LZero{}}}{\mathcal{W}^{all}}+(1-F_u)  \frac{\mathcal{W}_u^{\LOne{}}}{\mathcal{W}^{all}}\\
&\quad +\sum_{v \in \mathcal{N}(u) \setminus V^{\LZero{}} \setminus V^{\LOne{}}}(F_v-F_u) \frac{w(u v)}{\mathcal{W}^{all}} \notag
\end{split}
\end{equation}

Here, $\mathcal{W}^{all} = \sum_{v \in \mathcal{N}(u)} w(u,v)$ is the sum of the weights of all edges between $u$ and its neighbors. The second and third terms of the equation reflect the similarity contributions of $\LZero{}$ and $\LOne{}$, respectively. The last term reflects the contribution of the other neighbors of $u$.
If the difference between the updated label and the previous label of $u$ exceeds a predefined threshold $\delta$, its neighbors are likely to be affected and are flagged for label updates in the next iteration. The process converges when no such significant label changes occur in an iteration.




\color{black}

\begin{figure*}[htbp]
  \centering

  \begin{subfigure}[t]{0.245\textwidth}
    \includegraphics[width=\linewidth]%
      {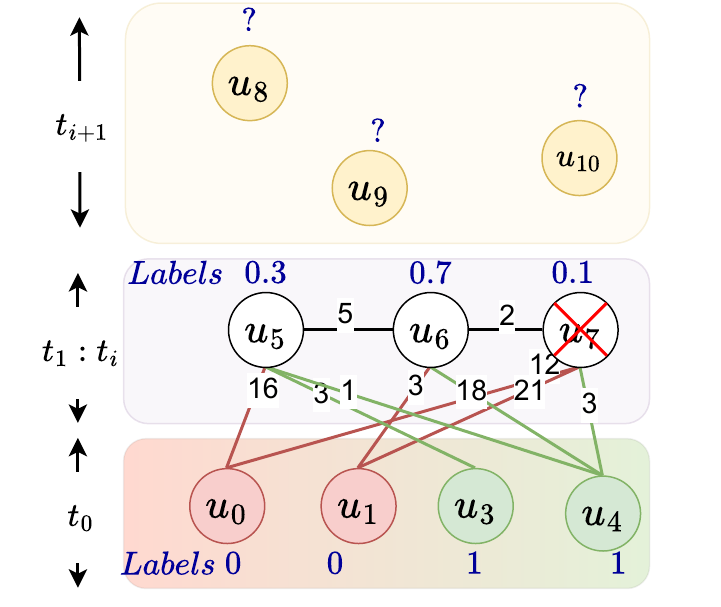}
    \caption{Initial state consists of ground truth at time $t_0$ and data came till time $t_i$}
    \label{fig:e1}
  \end{subfigure}
  \begin{subfigure}[t]{0.185\textwidth}
    \includegraphics[width=\linewidth]%
      {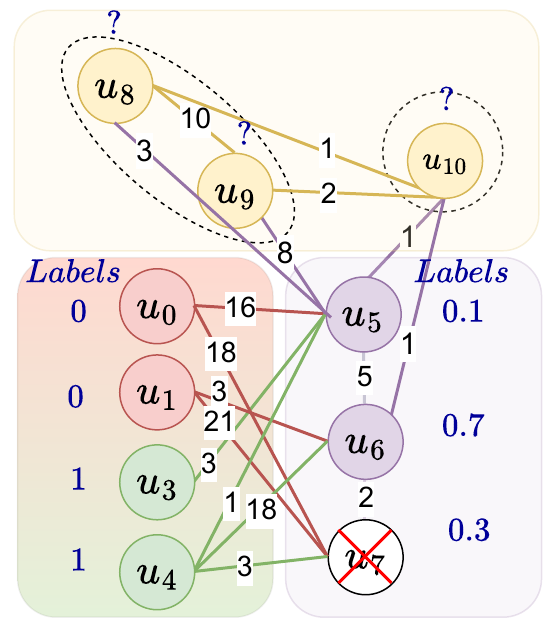}
    \caption{Arrival of new data at time $t_{i+1}$}
    \label{fig:e2}
  \end{subfigure}
  \begin{subfigure}[t]{0.27\textwidth}
    \includegraphics[width=\linewidth]%
      {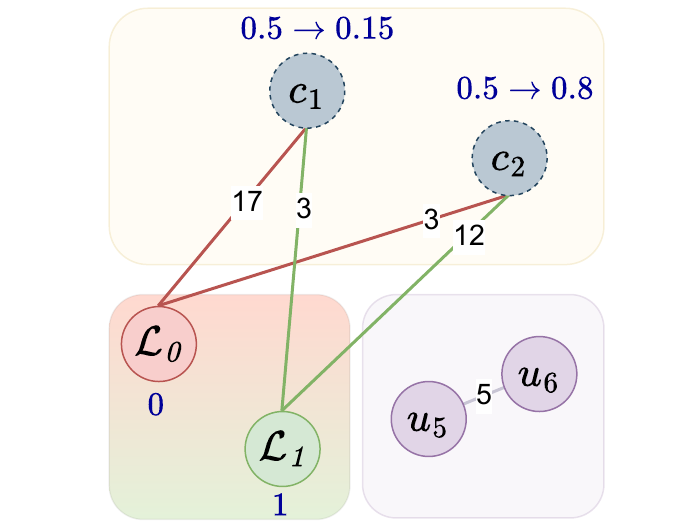}
    \caption{Compact representation with $u_0^*$ and $u_1^*$ as red and green class representative with parallel edge sum.}
    \label{fig:e3}
  \end{subfigure}
  \begin{subfigure}[t]{0.27\textwidth}
    \includegraphics[width=\linewidth]%
      {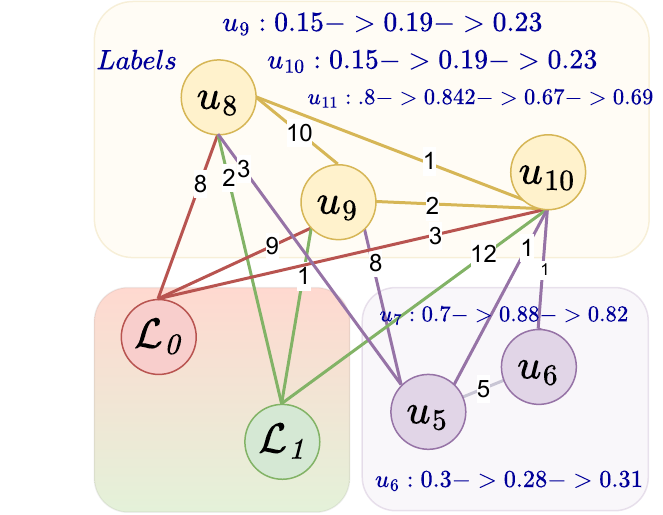}
    \caption{Each connected component is initialized as $0.5$ that converges after one iteration with updated value.} 
    \label{fig:e4}
  \end{subfigure}
  \caption{Evolution of the label-propagation algorithm over six iterations.}
  \label{fig:example}
  \vspace{-0.15in}
\end{figure*}

Figure~\ref{fig:example} provides an overview of the proposed Algorithm~\ref{alg:proposed}. Figure~\ref{fig:e1} depicts the initial setup, including the ground truth at time $t_0$ and all vertices observed during the interval $[t_1, t_{i+1}]$. The ground truth contains two classes: red indicates class label $0$, and green indicates class label $1$. Because the ground truth labels are fixed and do not change as new data points arrive, intermediate edges among ground truth vertices are omitted for clarity. Vertices appearing in the interval $[t_1, t_i]$ (purple region) are shown in color white  with their current estimated labels. Newly added vertices are shown in yellow, and deleted vertices are marked with a cross.
Step 1 computes the weights of edges associated with the new vertices and identifies the connected components. Figure~\ref{fig:e2} shows these connected components, indicated by dotted circles. To reduce visual clutter, we do not display edges formed between the new vertices and the ground truth vertices. Here, $u_5$ and $u_6$ are marked with color violet to indicate affected by their deleted or inserted neighbors $u_7, u_8, u_9, u_{10}$.
Figure~\ref{fig:e3} illustrates Step 2, where the labels of the new vertices in connected components $c_1$ and $c_2$ are initialized using the ground truth vertices. The supernodes composed of vertices with labels $0$ and $1$ are denoted by $\LZero{}$ and $\LOne{}$, respectively. 
Finally, Figure~\ref{fig:e4} illustrates the iterative label update (Step 3) for all affected vertices. The algorithm converges when no affected vertices remain. For example, the label of $u_9$ changes across iterations as $0.15 -> 0.19 -> 0.23$.

\section{Theoretical analysis}

\color{black}
\label{sec:theory}

In this section, we analyze the theoretical properties of the proposed iterative update rule used in \dblp{}. We first show that the update rule is equivalent to standard weighted neighborhood averaging, and then establish convergence to the unique harmonic solution by leveraging the convexity of the Dirichlet energy.

\paragraph{Equivalence between Iterative Update and Neighborhood Averaging.}
Let $u \in V^U$ be an unlabeled vertex with neighbor set $\mathcal{N}(u)$.
Assume that the labels $F_v$ of all neighbors $v \in \mathcal{N}(u)$ are fixed during the update of $u$.
Define the normalized edge weights
\[
\alpha_{u,v} = \frac{w(u,v)}{\sum_{x \in \mathcal{N}(u)} w(u,x)}, 
\qquad \forall v \in \mathcal{N}(u),
\]
which satisfy $\sum_{v \in \mathcal{N}(u)} \alpha_{u,v} = 1$.

\vspace{5pt}
\textbf{(1) Weighted Neighborhood Averaging.}
The classical label propagation update assigns to $u$ the weighted average of its neighbors’ labels:
\begin{equation}
\label{eq:avg}
F_u^\star = \sum_{v \in \mathcal{N}(u)} \alpha_{u,v} F_v .
\end{equation}

\textbf{(2) Iterative Adjustment Rule.}
The iterative update rule used in \dblp{} updates $F_u$ as
\begin{equation}
\label{eq:iter}
T(F_u) = F_u + \sum_{v \in \mathcal{N}(u)} \alpha_{u,v} (F_v - F_u).
\end{equation}

\textbf{(3) Equivalence.}
We show that the update in~\eqref{eq:iter} produces exactly the weighted average in~\eqref{eq:avg}, regardless of the current value of $F_u$:
\[
\begin{aligned}
T(F_u)
&= F_u + \sum_{v \in \mathcal{N}(u)} \alpha_{u,v} F_v
    - F_u \sum_{v \in \mathcal{N}(u)} \alpha_{u,v} \\
&= F_u + \sum_{v \in \mathcal{N}(u)} \alpha_{u,v} F_v - F_u \\
&= \sum_{v \in \mathcal{N}(u)} \alpha_{u,v} F_v \qquad = F_u^\star .
\end{aligned}
\]
Hence, the proposed iterative adjustment rule is equivalent to standard weighted neighborhood averaging and computes the local harmonic condition in a single update.

\paragraph{Convexity of the Dirichlet Energy.}
Let $G = (V,E)$ be a finite, connected, weighted graph with edge weights $w_{u,v} \ge 0$.
Let $V^L \subset V$ denote the set of vertices with ground truth labels, and $V^U = V \setminus V^L$ the unlabeled vertices.
For a label function $F : V \to \mathbb{R}$ satisfying $F_u = Y_u$ for all $u \in V^L$, define the Dirichlet energy
\begin{equation}
\label{eq:energy}
\mathcal{E}(F) = \frac{1}{2} \sum_{(u,v) \in E} w_{u,v} (F_u - F_v)^2 .
\end{equation}

\begin{lemma}[Strict Convexity]
\label{lem:convex}
The energy $\mathcal{E}(F)$ is strictly convex when restricted to the free variables $F_u$, $u \in V^U$.
\end{lemma}

\begin{proof}
Assume, for contradiction, that $\mathcal{E}(F)$ is not strictly convex on $V^U$.
Then there exist two distinct minimizers $F^{(1)} \neq F^{(2)}$ satisfying the boundary constraints on $V^L$ and
\[
\nabla \mathcal{E}(F^{(1)}) = \nabla \mathcal{E}(F^{(2)}) = 0 .
\]
This implies the existence of multiple harmonic extensions of the same boundary labels.

However, by the Dirichlet principle for finite graphs~\cite{zhu2003semi}, for a connected graph with a nonempty boundary set $V^L$, there exists a \emph{unique} harmonic function on $V^U$ that minimizes~\eqref{eq:energy}.
This contradicts the assumption of multiple minimizers.
Therefore, $\mathcal{E}(F)$ must be strictly convex on $V^U$.
\end{proof}

\paragraph{Convergence of the Iterative Update.}
We now establish convergence of the proposed update rule.

\begin{corollary}[Convergence to the Harmonic Solution]
\label{cor:convergence}
Let $G=(V,E)$ be a finite connected graph with ground truth vertices $V^L \neq \emptyset$.
If the labels of unlabeled vertices $u \in V^U$ are updated according to
\[
F_u \leftarrow F_u + \sum_{v \in \mathcal{N}(u)} \alpha_{u,v} (F_v - F_u),
\]
then the update process converges to the unique harmonic solution
\[
F^U = -\mathbf{L}_{UU}^{-1} \mathbf{L}_{UL} F^L .
\]
\end{corollary}

\begin{proof}[Proof Sketch]
From the equivalence established earlier, each update enforces the local harmonic condition.
By Lemma~\ref{lem:convex}, the Dirichlet energy has a unique minimizer over $V^U$.
Therefore, repeated application of the update rule converges to this unique harmonic solution, which coincides with the closed-form Laplacian-based solution.
\end{proof}

\color{black}

{\color{black}

\vspace{-0.1in}
\section{Implementation details}

\subsection{Load balancing}

We store the graphs in compressed sparse row (CSR) format where each vertex $v$ owns a row segment $[rowPtr[v],\,rowPtr[v+1])$ whose length $deg(v)$ may vary by orders of magnitude, leading to potential load imbalance. To handle this variability, we adopt a block-per-row segment execution model, where each thread block is assigned to process a single row segment. Threads within the block cooperatively traverse the neighbor list in a block-strided manner, enabling efficient and parallel processing of high-degree vertices.
Partial results computed by individual threads are combined using shared-memory block-level reduction. Although the degree of row segments remain irregular, this cooperative block-level parallelism mitigates long-tail effects for hub vertices and allows the GPU scheduler to maintain overall load balance through concurrent execution of multiple blocks.

\subsection{Kernel design}
Figures~\ref{fig:kernel} and~\ref{fig:dynamic} illustrate the GPU kernel designs developed to achieve scalability.

\color{black}
Given a graph where each pairwise similarity between vertices are available, we obtain the connected components by temporarily removing all edges whose weights are below a threshold $\tau$. Figure~\ref{fig:k1} presents the kernel responsible for temporarily removing edges by negating the destination vertex ID in the col array of the CSR data structure. This operation does not actually remove the edge from the structure; instead, it flags it for exclusion in later steps where edge information may still be needed. For instance, with a similarity threshold value of $\tau = 5$, the kernel marks entries at indices ${1, 3, 4, 5}$ in the col array as deleted. This task is embarrassingly parallel and benefits from memory coalescing, a performance optimization enabled by the regular access patterns of the CSR format.
\color{black}


For connected component find, we adopt the\\ Shiloach–Vishkin (SV) algorithm~\cite{vishkin1982log} 
that relies on simple, data-parallel operations, and its pointer-jumping step involves parent updates through regular array accesses, making it well-suited for GPU execution. 
Furthermore, mapping one thread per vertex enables high occupancy, effectively hiding memory latency and improving overall throughput. 
Figure~\ref{fig:k2} illustrates the kernel implementing SV algorithm, which consists of two iterative steps: (i) Hook and (ii) Jump. These kernels continue alternating until no changes are detected after a Jump step, indicating convergence.

In the Hook phase, each thread is assigned to a vertex and updates its parent in the \textit{par }array to be the minimum of its current parent and the IDs of its adjacent vertices (including itself). For example, as shown in Figure~\ref{fig:k2} (left), vertex index $9$ updates its parent from $9$ to $8$ (since $u_{8}$ has the smallest ID among its neighbors). Vertices $u_{8}$ and $u_{10}$ remain unchanged.

The Jump phase then compresses the paths in the parent array by updating each element par[i] to par[par[i]], effectively halving the distance to the root. This process is illustrated in Figure~\ref{fig:k4}. In this example, as the Jump phase makes no further updates, the iteration terminates after a single pass. The final \textit{par }array contains two unique parent IDs ${0, 2}$, indicating the existence of two connected components. However, these component IDs are non-sequential (e.g., ID $1$ is skipped), which is resolved via prefix scan operation using thrust library.

Figure~\ref{fig:k3} demonstrates the same kernel logic applied with a lower threshold, $\tau = 1$, resulting in fewer edges being marked for deletion. Figure~\ref{fig:k4} shows the Hook kernel producing a parent array of ${8, 8, 9}$, where $u_{10}$ identifies $u_{9}$ as its parent. The subsequent Jump phase updates $u_{10}$ to have the same parent as $u_{9}$, effectively compressing the path. Since a change occurred, the algorithm proceeds with further iterations of Hook and Jump until no updates are detected in the Jump step.


\begin{figure*}[htbp]
  \centering

  \begin{subfigure}[t]{0.20\textwidth}
    \includegraphics[width=\linewidth]%
      {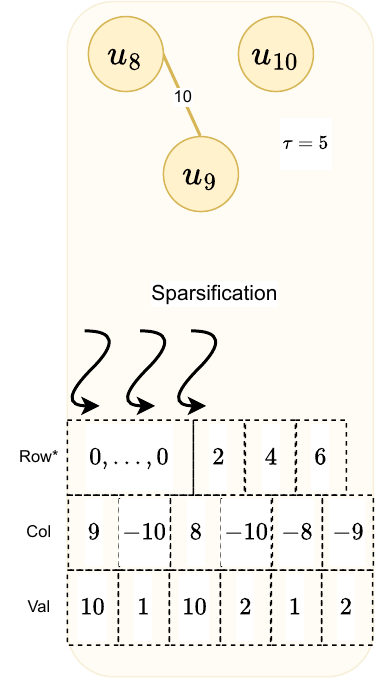}
    \caption{Parallel sparsification kernel step the node as negative to denote as temporary deleted for edge weight less than or equal to $\tau = 5$}
    \label{fig:k1}
  \end{subfigure}\hfill
  \begin{subfigure}[t]{0.22\textwidth}
    \includegraphics[width=\linewidth]%
      {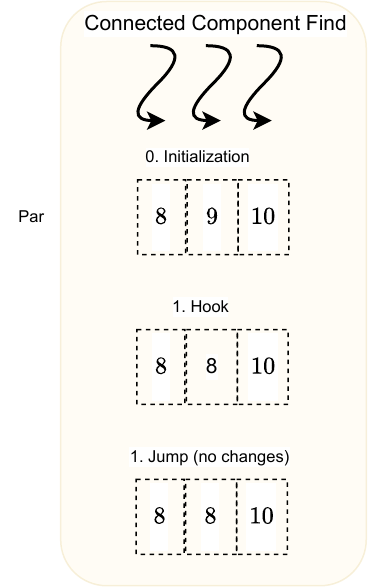}
    \caption{Parallel connected component find for (a)}
    \label{fig:k2}
  \end{subfigure}\hfill
  \begin{subfigure}[t]{0.20\textwidth}
    \includegraphics[width=\linewidth]%
      {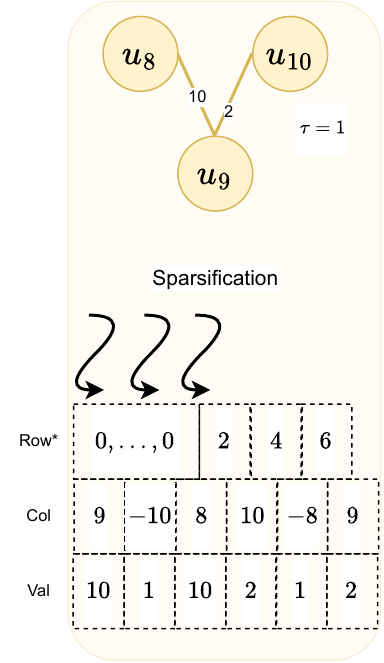}
    \caption{Parallel sparsification kernel sets the node as negative to denote as temporary deleted for edge weight less than or equal to $\tau = 1$}
    \label{fig:k3}
  \end{subfigure}\hfill
  \begin{subfigure}[t]{0.15\textwidth}
    \includegraphics[width=\linewidth]%
      {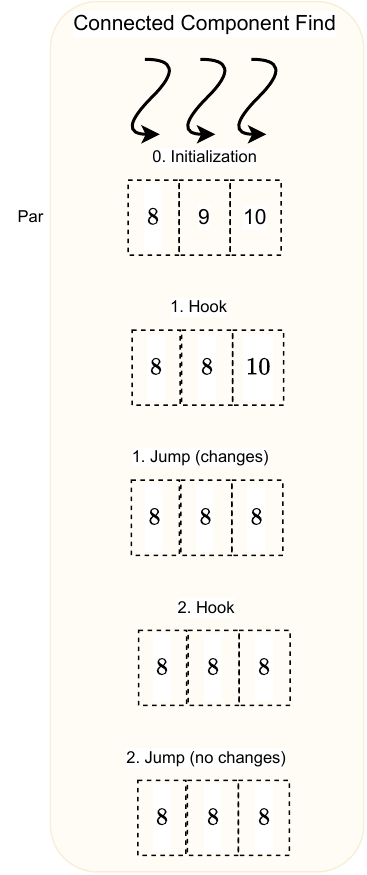}
    \caption{Parallel connected component find for (c)}
    \label{fig:k4}
  \end{subfigure}

  \caption{CUDA Kernel design for sparsification and connected component finding}
  \label{fig:kernel}
  \vspace{-0.1in}
\end{figure*}

Computing the parallel edge sum is a critical operation used in both Line 22 and Line 28 of Algorithm~\ref{alg:proposed}. To achieve efficient parallelism, we aim to process all vertices in the affected node set (i.e., the frontier list) concurrently. However, each vertex must aggregate edge weights over three distinct subsets of nodes: (i) the ground truth nodes, (ii) vertices that appeared from time $t_1$ to $i$, and (iii) vertices appearing at time $i+1$. 

Assigning a single thread per vertex would result in sequential edge sum computations within each vertex’s neighborhood, introducing significant performance bottlenecks. To overcome this limitation, we employ CUDA block level parallelism. Instead of mapping a single thread to each vertex in the affected set, we launch an entire thread block per vertex. Within each block, multiple threads collaboratively compute the edge sum in parallel, significantly reducing computation time. This approach is illustrated in Figure~\ref{fig:dynamic}.

In cases where the number of neighboring vertices exceeds the number of threads in a block, we implement strided parallelism. 

\begin{figure}[htbp]
  \centering
  \includegraphics[width=0.8\linewidth]{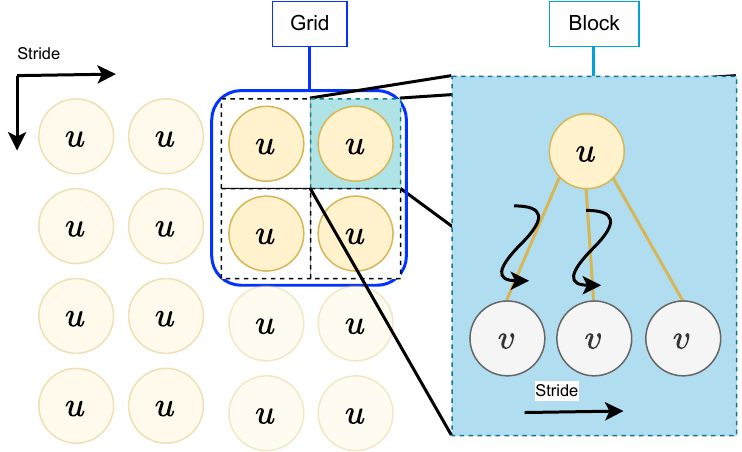}
  \caption{Block level granularity to process nodes in the frontier list }
  \label{fig:dynamic}
\end{figure}

\begin{figure*}[htbp]
  \centering

  \begin{subfigure}[t]{0.30\textwidth}
    \includegraphics[width=\linewidth]%
      {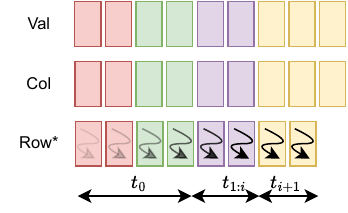}
    \caption{CSR vector layout in host device memory to update graph efficiently}
    \label{fig:m1}
  \end{subfigure}\hfill
  \begin{subfigure}[t]{0.30\textwidth}
    \includegraphics[width=\linewidth]%
      {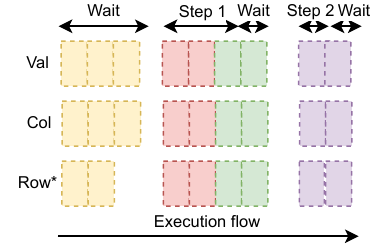}
    \caption{Asynchronous memory transfer and overlapping kernel execution to hide memory latency}
    \label{fig:m2}
  \end{subfigure}\hfill
  \begin{subfigure}[t]{0.25\textwidth}
    \includegraphics[width=\linewidth]%
      {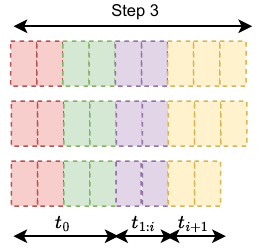}
    \caption{GPU memory layout}
    \label{fig:m3}
  \end{subfigure}

  \caption{Memory latency hiding}
  \label{fig:memory}
  \vspace{-0.2in}
\end{figure*}

\subsection{Overlapping subgraph transfer and \\kernel execution}
Figure~\ref{fig:memory} illustrates the use of asynchronous memory transfer from host to device for updating the CSR graph stored in host (CPU) memory. Rather than transferring the entire graph to the GPU for each incoming batch of vertices, we asynchronously transfer only the required portions. This approach enables overlapping memory transfer with kernel execution, effectively hiding memory transfer latency and improving overall throughput.

 Figure~\ref{fig:m1} depicts the parallel update of the graph in CPU memory. On the CPU, we represent the graph using a 2-D vector structure to facilitate dynamic memory allocation. Since the graph grows incrementally—only allowing additions of vertices and edges—we parallelize the update process by assigning each incoming batch of vertices to separate threads. This enables efficient concurrent insertion of new edges and vertices without the need for global synchronization.

Subsequently, we transfer data from the CPU to the GPU using asynchronous memory transfer, as illustrated in Figure~\ref{fig:m2}. Instead of transferring the entire graph, we begin by transferring only the vertices arriving at time $t_{i+1}$, highlighted in yellow. As soon as this memory transfer is initiated, we launch the sparsification and connected component kernels (corresponding to Step 1 and Step 2 in Algorithm~\ref{alg:proposed}) on this subset. Concurrently, while the sparsification kernel is executing, we initiate the transfer of ground truth data (shown in red and green) to the GPU. Once the ground truth data is available on the device, we perform the parallel edge sum operation as described in Line 14 of Algorithm~\ref{alg:proposed}. Additionally, the memory corresponding to vertices arriving at $t_{1:i}$ is transferred immediately after the ground truth at $t_0$ is available, enabling the execution of Step 3. This pipelined data transfer and execution strategy significantly reduces idle time and improves throughput by an average factor of $9.7\times$.

Figure~\ref{fig:m3} illustrates the static memory layout of the graph in GPU memory. Although the layout itself remains fixed, the rows ($\mathit{row}^*$) corresponding to vertices arriving at times $t_1 \dots t_{i+1}$ are appended sequentially as they arrive. This contiguous storage pattern ensures that each kernel benefits from coalesced memory access, thereby improving memory bandwidth utilization and overall performance.

}

\begin{table}
\caption{Baseline methods compared with \dblp{}.}
    \label{tab:baselines}
    \centering
    \small
    \begin{tabular}{|c|cccccc|}\hline
        Method  & \makecell{\rotatebox{90}{\textbf{Incremental}}} & \makecell{\rotatebox{90}{\textbf{Decremental}}}  &
        \makecell{\rotatebox{90}{\textbf{Parallelism}}} & \makecell{\rotatebox{90}{\textbf{Update}}} &
        \makecell{\rotatebox{90}{\textbf{Target graph}}} &
        \makecell{\rotatebox{90}{\textbf{Memory Optim.}}}\\ \hline
         \itlp{}~\cite{zhu2003semi} & \cmark &   \xmark &  \xmark &  \xmark & Dense & Compression\\
        CAGNN\cite{zhu2020cagnn} & \cmark &   \xmark &  \cmark  &  \xmark & Dense & Compression\\
        A2LP\cite{zhang2020label} & \cmark &   \xmark &  \cmark  &  \xmark & Sparse & KNN\\
        \stlp{}~\cite{wagner2018semi} & \cmark & \cmark & \cmark  &  \xmark & Dense & N/A\\
        Approx \stlp{}\cite{ponte2024computing} & \cmark & \cmark & \cmark  &  \xmark & Dense & Sparse Inverse\\
        \dblp{} [this work] & \cmark & \cmark & \cmark & \cmark & Sparse & CSR \\ \hline
    \end{tabular}
    \vspace{-0.15in}
\end{table}

\section{Experimental Results}

\subsection{Experimental Setup}

All GPU experiments were conducted on an NVIDIA H100 with 80 GB of VRAM. 
The host machine is equipped with an AMD EPYC 7502 32-Core CPU and 32 GB of RAM.

\textit{Baselines. } We compare \dblp{} with the state-of-the-art methods listed in Table~\ref{tab:baselines}. We implement the average-neighborhood method \itlp{} in a GPU setting to evaluate both iteration count and parallel execution time. We also efficiently parallelize \stlp{}, which was traditionally constrained by the inverse adjacency matrix, leading to high execution time and memory usage. We mitigate this limitation by incorporating an approximate inverse\cite{ponte2024computing}, improving memory scalability. In addition, we include machine-learning baselines such as A2LP\cite{zhang2020label} and CAGNN\cite{zhu2020cagnn} to provide a comprehensive comparison 
across approaches.

\textit{Datasets. } We evaluate \dblp{} using synthetic and real-world datasets listed in Table~\ref{tab:dataset_description}. The synthetic sparse graphs are generated using the Erdős–Rényi model, with the average degree varying among 3, 5, and 7. 
Following the standard approach~\cite{subramanya2014graph}, non-graph datasets are modeled as graphs. For ImageNet, we select an image set of 50K with classes “cat” and “non-cat” to form a balanced binary classification problem. Each image is represented as a node, and feature vectors were extracted from the \emph{penultimate layer} of a pretrained model ResNet-50~\cite{he2016deep}. Pairwise cosine similarities~\cite{bayardo2007scaling} are computed to construct a fully connected similarity matrix, which is then sparsified using a $k$-nearest neighbor (kNN) with $k=5$ as proposed in \cite{lingam2021glam}.
\begin{table}
\caption{Dataset description}
    \label{tab:dataset_description}
    \small
    \centering
    \begin{tabular}{|c|c|c|} \hline
        Dataset & $|V|$ & $|E|$\\ \hline
        IMDB Review\cite{maas-EtAl:2011:ACL-HLT2011} & 50,000 & 125K  \\
        ImageNet\cite{deng2009imagenet} & 50,000 & 125K \\
        Yelp Review\cite{yelp_review_dataset_kaggle} & 6,990,280 & 17M\\
        Amazon Household Review\cite{hou2024bridging} & 25,600,000 & 64M\\
        Amazon Book Review\cite{hou2024bridging} & 29,500,000 & 73M\\
        Random Dataset &  50,000,000 & \{75M,62M,175M\} \\ \hline
    \end{tabular}
    \vspace{-0.2in}
\end{table}
For the review datasets (IMDB, Yelp, Amazon Household, and Amazon Book), each review is considered as a vertex. We compute Term Frequency-Inverse Document Frequency (TF–IDF) \cite{sparck1972statistical} of the reviews to generate embeddings, then apply pairwise cosine similarity among them and sparsify using the aforementioned kNN-based strategy to form edges. The IMDB dataset is inherently binary-labeled, while for Yelp and Amazon reviews, we convert the star ratings into binary classes, assigning label $1$ to reviews with a rating of three stars or higher, and $0$ otherwise.

 The \itlp{}, \stlp{} and \dblp{} support dynamic updates, including the insertion of ground-truth and unlabeled vertices, as well as the deletion of existing vertices. In all experiments, each batch of changes consists of $90\%$ unlabeled new vertices, $1\%$ vertices with ground-truth, and $9\%$ deleted vertices. For deletions, vertices are randomly selected from the subgraph of the existing graph while ensuring that the entire graph is not removed. When the required number of deletions exceeds the number of available vertices, sampling is performed with replacement, allowing the same vertex to be selected multiple times to avoid size inconsistencies.
 
\vspace{-0.5em}

\begin{figure}[htbp]
  \centering

  \begin{subfigure}[b]{0.45\linewidth}
    \centering
    \includegraphics[width=\linewidth]{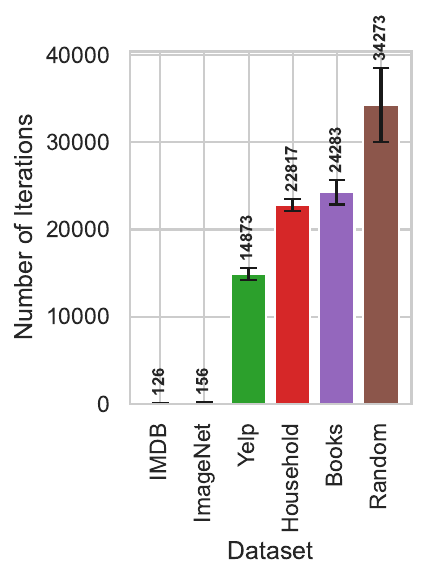}
    \caption{Required iterations.}
    \label{fig:iterations}
  \end{subfigure}
  \hfill
  \begin{subfigure}[b]{0.45\linewidth}
    \centering
    \includegraphics[width=\linewidth]{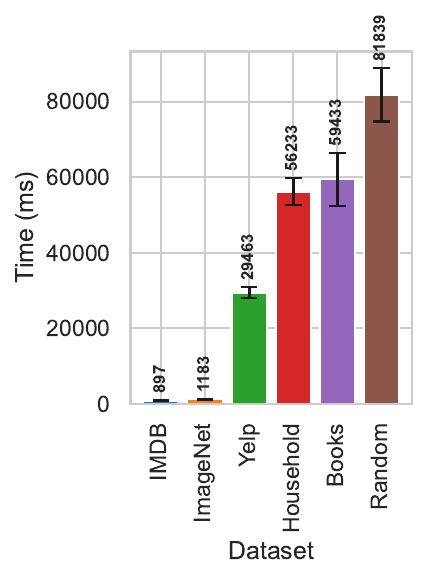}
    \caption{Execution time.}
    \label{fig:execution}
  \end{subfigure}

  \caption{Iterations and execution time of \dblp{} across datasets.}
  \label{fig:prop}
\vspace{-0.15in}
\end{figure}

\vspace{-0.5em}
\subsection{Experiment on \dblp{} properties}

In our first experiment, we set the average degree of all datasets to 5 with the goal to study the impact of the size of the datasets for \dblp{}. We assume that $1\%$ of the vertices in each dataset have ground truth labels, while the remaining vertices require labeling. We place all unlabeled vertices into a single batch and process them using \dblp{}. Figure~\ref{fig:iterations} and Figure~\ref{fig:execution} report the required iterations and execution time of \dblp{} across datasets, respectively. We observe that as the number of vertices increases, both the iterations and the execution time increase. IMDB, as the smallest dataset, requires the fewest iterations (average 126) and the least time (average 897 ms). In contrast, the random graph with $1000\times$ more vertices than IMDB requires 34,273 iterations and has an average execution time of 81,839 ms.

\begin{figure}[htbp]
  \centering
  \begin{subfigure}[b]{0.88\linewidth}
    \centering
    \includegraphics[width=\linewidth]{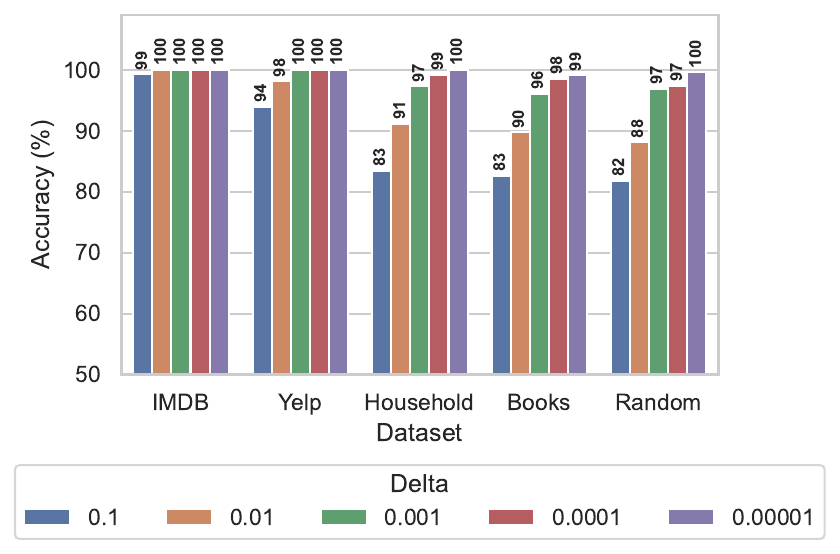}
    \caption{Accuracy variation on different datasets.}
    \label{fig:delta1}
  \end{subfigure}

  \begin{subfigure}[b]{0.88\linewidth}
    \centering
    \includegraphics[width=\linewidth]{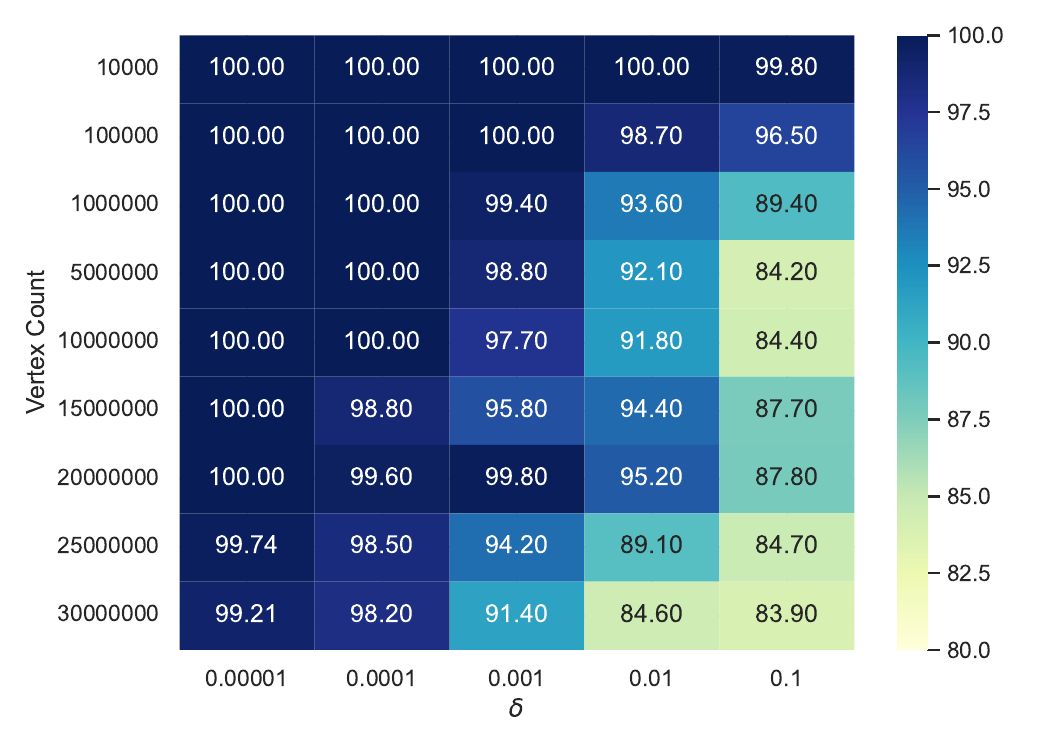}
    \caption{Accuracy variation on different batchsize.}
    \label{fig:delta2}
  \end{subfigure}

  \caption{Impact of $\delta$ on \dblp{}}
  \label{fig:prop}
\vspace{-0.25in}
\end{figure}


In the next set of experiments, we vary the update threshold $\delta$ from the set $\{0.1, 0.01, 0.001, 0.0001, 0.00001\}$ and study its impact on the execution of \dblp{}. We find that $\delta$ directly affects the number of iterations, which in turn determines the total execution time. As $\delta$ increases, Algorithm~\ref{alg:proposed} terminates faster because Step 3 requires fewer iterations. However, early termination can reduce accuracy.
Here, by accuracy, we mean the fraction of correctly predicted levels divide by total predicted levels. Each of the levels is mapped to either 0 or 1 with a cutoff probability threshold of 0.5. The accuracy is measured relative to the baseline method of Wagner et al.~\cite{wagner2018semi}, which optimally minimizes the energy function.

Figure~\ref{fig:delta1} shows that accuracy is lower for larger $\delta$ and in most cases, $\delta = 0.0001$ achieves near optimal accuracy. Reducing it further slightly improves accuracy, but increases the number of iterations and the execution time.
We also observe that accuracy decreases as the graph size increases. To further analyze the impact of $\delta$ under different input batch sizes, we vary the number of vertices per batch from 10,000 to 30,000,000 on a random graph. Figure~\ref{fig:delta2} shows that accuracy decreases as batch size increases. We find that $\delta = 0.0001$ or smaller achieves near optimal accuracy across all batch sizes. Therefore, we use $\delta = 0.0001$ in the subsequent experiments.

\subsection{Comparison with baselines}

\begin{figure*}[htbp]
  \centering
  \begin{subfigure}[b]{0.30\textwidth}
    \centering
    \includegraphics[width=\linewidth]{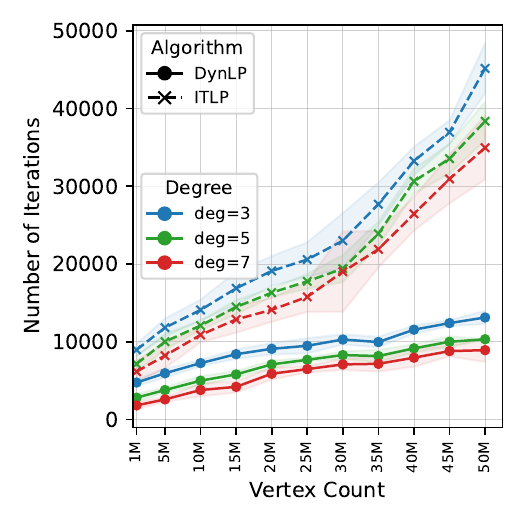}
    \caption{Iterations (Random Graph)}
  \end{subfigure}
  \hfill
  \begin{subfigure}[b]{0.30\textwidth}
    \centering
    \includegraphics[width=\linewidth]{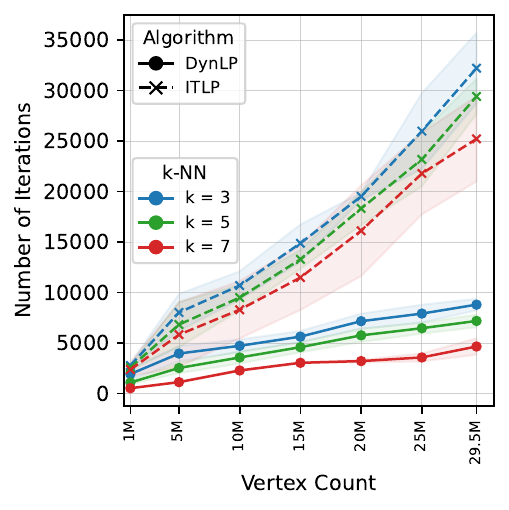}
    \caption{Iterations (Amazon Books)}
  \end{subfigure}
    \hfill
    \begin{subfigure}[b]{0.30\textwidth}
    \centering
    \includegraphics[width=\linewidth]{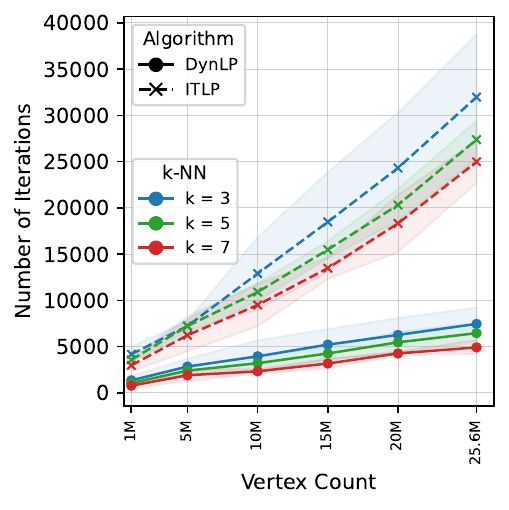}
    \caption{Iterations (Amazon Household)}
  \end{subfigure}
  \hfill
  
  \begin{subfigure}[b]{0.30\textwidth}
    \centering
    \includegraphics[width=\linewidth]{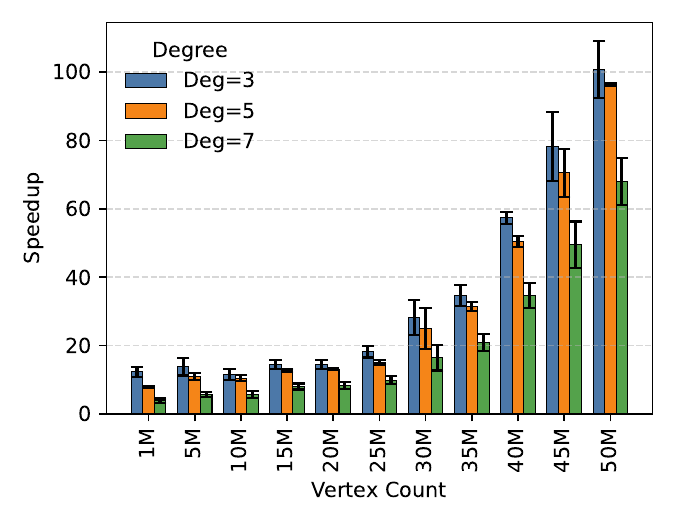}
    \caption{Speedup (Random Graph)}
  \end{subfigure}
  \hfill
  \begin{subfigure}[b]{0.30\textwidth}
    \centering
    \includegraphics[width=\linewidth]{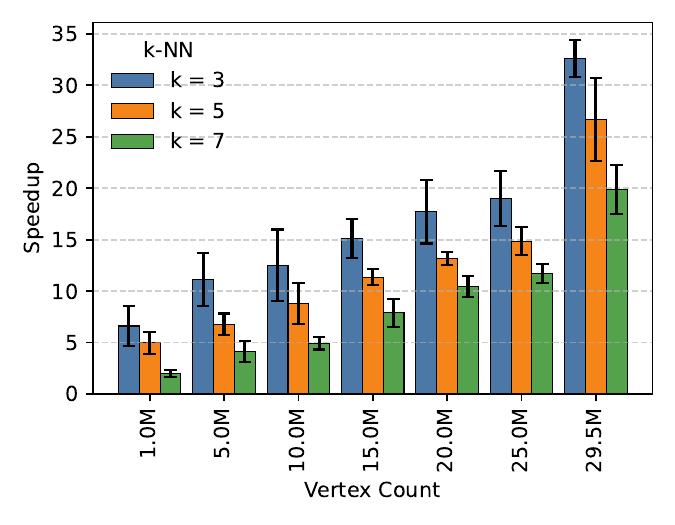}
    \caption{Speedup (Amazon Books)}
  \end{subfigure}
  \hfill
  \begin{subfigure}[b]{0.30\textwidth}
    \centering
    \includegraphics[width=\linewidth]{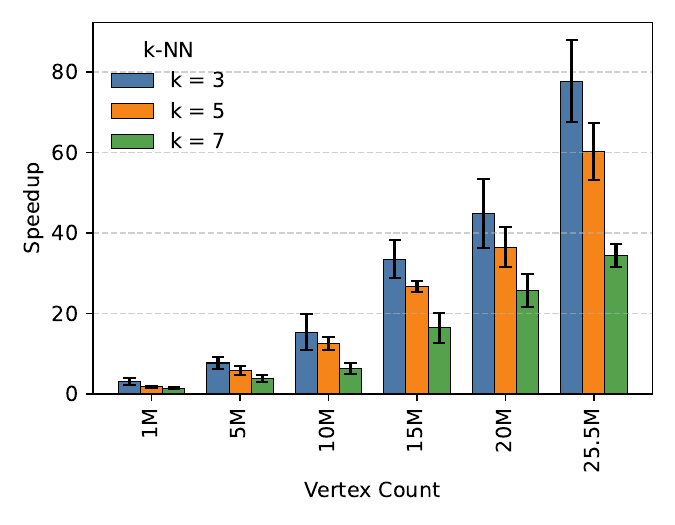}
    \caption{Speedup (Amazon Household)}
  \end{subfigure}

  \caption{Iteration and speedup comparison between \dblp{} and \itlp{} as vertex count varies. }
  \label{fig:iter_speedup_all}
  \vspace{-0.15in}
\end{figure*}




\textbf{Comparison with \itlp{}: } We compare \dblp{} with our GPU implementation of \itlp{}. For each graph, we begin with $10^4$ randomly selected initial vertices with ground truth labels. Then, in each batch, $5$ million vertices are added, and this process continues until the total number of vertices matches that of the original graph. In this experiment, we vary the average degree (for the random graph) and $k$ (for the non-graph data) among $3,5$, and $7$. Since both \dblp{} and \itlp{} rely on iterative convergence, we compare their required number of iterations in Figures~\ref{fig:iter_speedup_all}(a), (b), and (c). Note that we set the convergence parameter $\delta = 0.0001$ for both \dblp{} and \itlp{}. We observe that across all experiments, \itlp{} requires more iterations than \dblp{} because \itlp{} recomputes labels for all vertices in every round, whereas our method updates labels for previously inserted vertices and efficiently computes labels for newly added vertices using a connected component assisted label initialization technique. Moreover, the gap in required iterations increases as the total vertex count grows. We also find that, for both methods, the iteration count decreases as the average degree increases. With the number of vertices fixed, increasing the average degree (or $k$) makes the graph denser, which reduces the hop distance between vertices. Consequently, in denser graphs, labels propagate to more vertices within fewer iterations, reducing the overall iteration requirement.

Figures~\ref{fig:iter_speedup_all}(d), (e), and (f) plot the speedup, computed as the ratio of the execution time of \itlp{} to that of \dblp{}. On the random graph, \dblp{} runs up to $100\times$ faster than \itlp{}. On the Amazon books and household datasets, \dblp{} achieves up to $35\times$ and $80\times$ speedup, respectively. Consistent with the iteration trends, the speedup increases as the graph’s average degree decreases.

\begin{figure}[htbp]
\vspace{-0.05in}
  \centering
  \begin{subfigure}[b]{0.48\linewidth}
    \centering
    \includegraphics[width=\linewidth]{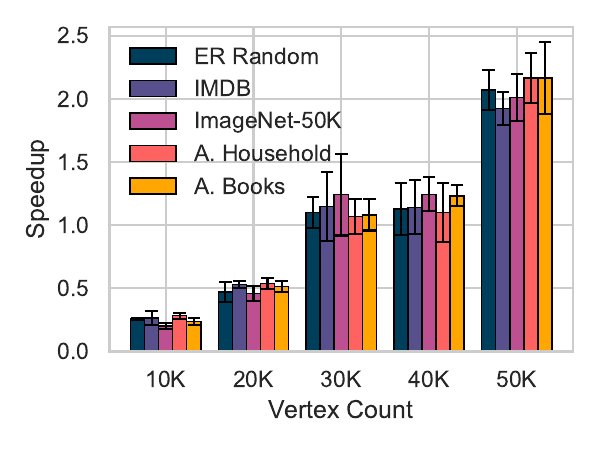}
    \caption{Speedup over the \stlp{} considering only kernel execution time.}
    \label{fig:speedupkernel}
  \end{subfigure}
  \hfill
  \begin{subfigure}[b]{0.48\linewidth}
    \centering
    \includegraphics[width=\linewidth]{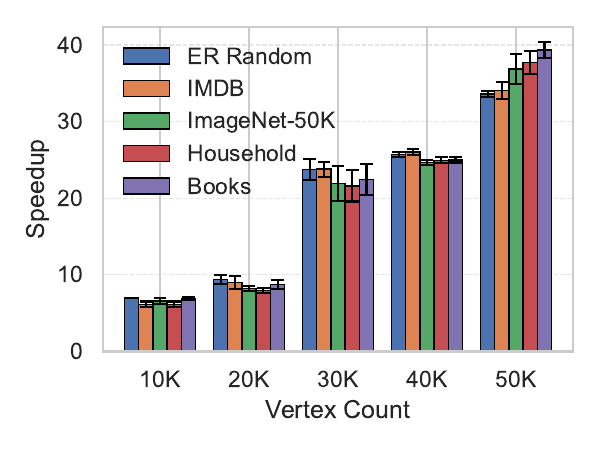}
    \caption{Speedup considering memory transfer + kernel execution time.}
    \label{fig:speedupkernelmem}
  \end{subfigure}
  \caption{Speedup comparison between \dblp{} and \stlp{}}
  \label{fig:iter_speedup_combined_Wagner}
\vspace{-0.15in}
\end{figure}



\textbf{Comparison with \stlp{}: }
Here, we compare our proposed \dblp{} with a GPU implementation of \stlp{}~\cite{wagner2018semi}.
Due to the $O(n^2)$ space complexity of \stlp{},
we were able to test the baseline only up to 50,000 nodes. Although we used a sparse graph, the Laplacian matrix of such a graph, required for \stlp{}, still exhibits quadratic space complexity, which severely limited scalability.
Figure~\ref{fig:speedupkernel} illustrates the kernel-level speedup of our proposed algorithm relative to the baseline. Initially, our method incurs overhead from the connected component find step, but this cost is quickly amortized as the batch size increases. The primary bottleneck in the baseline lies in its repeated matrix inversion and harmonic solution recomputation for each batch, which dominates its execution time.

When memory transfer time between host and device is also considered in the speedup computation, as shown in Figure~\ref{fig:speedupkernelmem}, the performance gap widens further. Our algorithm employs a CSR-based data structure, which is highly efficient for sparse graphs, while the baseline implementation constructs the Laplacian matrix without exploiting the benefits of sparsity, resulting in significantly higher memory and computational costs.

\begin{figure}[htbp]
  \centering
  \begin{subfigure}[b]{0.48\linewidth}
    \centering
    \includegraphics[width=\linewidth]{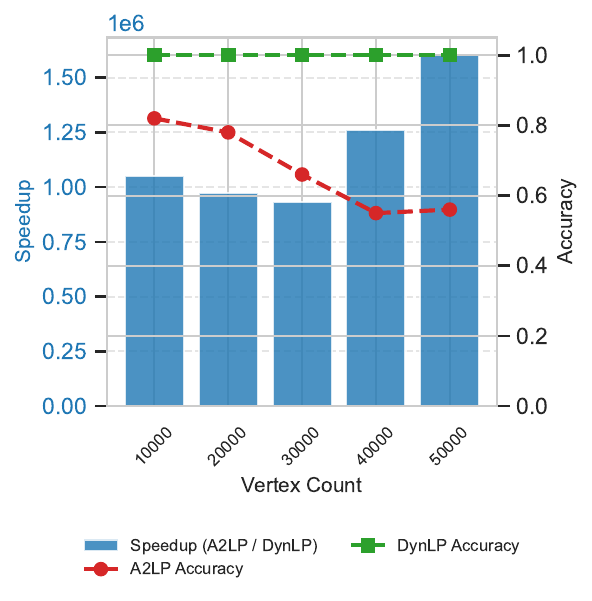}
    \caption{Comparison with A2LP}
    \label{fig:A2LB vs DBLB on Random Graph}
  \end{subfigure}
  \hfill
  \begin{subfigure}[b]{0.45\linewidth}
    \centering
    \includegraphics[width=\linewidth]{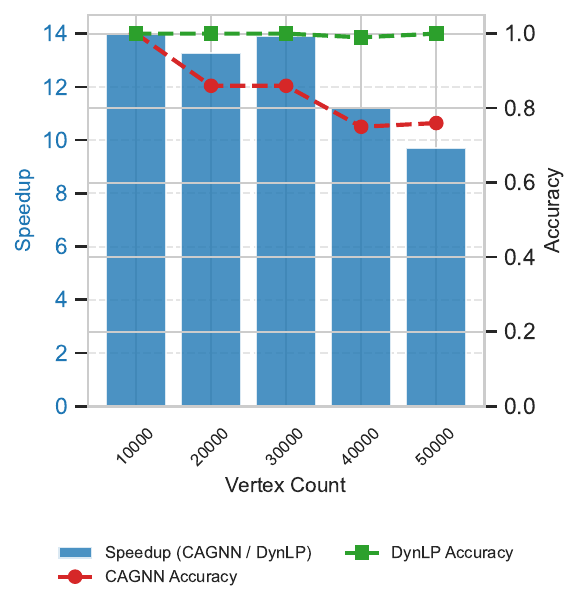}
    \caption{Comparison with CAGNN}
    \label{fig:CAGNN vs DBLB on IMDB Graph}
  \end{subfigure}
  \caption{Performance comparison with machine learning methods}
  \vspace{-0.1in}
\end{figure}

\textbf{Comparison with machine learning-based approaches: }

As A2LP is a convolutional neural network-based approach, it is best suited to image datasets. We use 50,000 ImageNet samples, modeled as vertices, to evaluate both \dblp{} and A2LP. From each class, 1,000 nodes are randomly selected as labeled ground truth, and the remaining nodes are divided into batches of approximately 10,000 samples for incremental updates. Figure~\ref{fig:A2LB vs DBLB on Random Graph} shows that \dblp{} achieves, on average, a $10^6\times$ speedup over A2LP. In terms of accuracy, A2LP reaches $80\%$ and its accuracy decreases as the total number of vertices increases.
We also compare \dblp{} with CAGNN, a two-layer Graph Convolutional Network (GCN) architecture configured with \texttt{SVD\_DIM}=512, \texttt{HIDDEN\_DIM}=256, and \texttt{OUT\_DIM}=2. CAGNN is more scalable than A2LP, and in addition to ImageNet and IMDB, it also runs on larger datasets such as Yelp. Figure~\ref{fig:CAGNN vs DBLB on IMDB Graph} shows the speedup of \dblp{} over CAGNN on IMDB data and compares accuracy. We observe that \dblp{} achieves up to a $14\times$ speedup. Compared to \dblp{}, CAGNN achieves $100\%$ accuracy when the total vertex count is small; however, its accuracy decreases as the number of vertices increases.

\begin{table}[h!]
\caption{Execution time comparison across datasets.}
\label{tab:execution}
\centering
\small
\begin{tabular}{|l|r|r|r|r|}
\hline
\textbf{Dataset} & \textbf{\itlp{}} & \textbf{\stlp{}} & \textbf{\dblp{} } & {CAGNN}  \\ \hline
IMDB & 652.25 & 3,196.41 & 439.21 & 8,545.53\\
Yelp & 18,283.94 & 23,182.12 ($\gamma = 10$) & 3,712.82 & 40,853.31\\
 Household & 678,939.29 & - & 18,232.38 & - \\
 Book & 783,925.09 & - & 19,927.31 & - \\
\hline
\end{tabular}

\end{table}

\vspace{-0.5em}
\textbf{More on execution time: }
Table~\ref{tab:execution} compares execution times of the proposed algorithm with baselines on the IMDB, Yelp, Amazon Household, and Book Review datasets. 
All execution times correspond to processing a single batch with $1\%$ initial ground-truth labels. \dblp{} outperforms all baselines, and the performance gap widens as graph size increases. For \stlp{}, the primary bottleneck is memory, which restricts its execution to the smallest dataset, IMDB. Using the approximation method proposed in \cite{ponte2024computing} with $\gamma = 10$, \stlp{} can also run on Yelp. Here, $\gamma$ is a parameter introduced in \cite{ponte2024computing} to control the trade-off between sparsity and approximation quality of the matrix inverse. A larger $\gamma$ promotes a sparser generalized inverse, but may lead to a poorer approximation. In contrast, a smaller $\gamma$ keeps the solution closer to the Moore-Penrose inverse, at the cost of reduced sparsity and higher memory usage\cite{ponte2024computing}.

\begin{table}[htbp]
\vspace{-0.05in}
\centering
\caption{Performance comparison on random graph with varied batch sizes (T: Execution time, A: Accuracy)}
\begin{tabular}{|c|cc|cc|cc|}
\hline
\textbf{Method} 
& \multicolumn{2}{c|}{\textbf{50K}} 
& \multicolumn{2}{c|}{\textbf{500K}} 
& \multicolumn{2}{c|}{\textbf{5M}} \\
& T(ms) & A & T(ms) & A & T(ms) & A \\
\hline
\itlp{} 
& 1,120 & 100 
& 4,738 & 100 
& 11,929 & 100 \\
\hline
\stlp{} 
& 1,637 & 100 
& -- & -- 
& -- & -- \\
\hline
\stlp{}($\gamma = 0.1$) 
& 5,637 & 72.9 
& -- & -- 
& -- & -- \\
\hline
\stlp{} ($\gamma = 1.0$) 
& 3,989 & 83.5 
& -- & -- 
& -- & -- \\
\hline
\stlp{} ($\gamma = 10.0$) 
& 1,563 & 56.3 
& 5,989 & 54.2 
& 21,637 & 49.5 \\
\hline
CAGNN 
& 7,637 & 100 
& 31,293 & 96 
& 92,838 & 88 \\
\hline
A2LP 
& 7,637 & 82 
& -- & -- 
& -- & -- \\
\hline
\dblp{} 
& 473 & 100 
& 2,128 & 99.3 
& 3,271 & 97.9 \\
\hline
\end{tabular}
\label{tab:epochvsiter}
\end{table}

\vspace{-0.5em}
\textbf{Memory and accuracy:} In our experimental setup, we evaluated different categories of label propagation methods for an exhaustive comparison. However, not all baselines are equally scalable in terms of memory. Due to memory limitations, we varied the single-batch size to highlight their differences. As shown in Table~\ref{tab:epochvsiter}, the \stlp{} method is restricted to a batch size of 50K because of its quadratic memory requirement. This limitation can be partially alleviated by using an approximate matrix inverse, allowing the batch size to scale up to 5M; however, this comes at a significant loss in accuracy.

For A2LP, the limited availability of ground-truth labels prevents the learning model from generalizing effectively, even for batch sizes of 50K. Increasing the batch size further degrades its performance, making it comparable to random binary classification. 

In contrast, the methods that scale well in practice, such as \itlp{} and CAGNN, demonstrate better memory efficiency. Compared to these approaches, our method achieves lower execution time by enabling efficient initialization and avoiding redundant computations, while maintaining accuracy close to the optimal solution.
\color{black}

\textbf{Comparison Summary:}
Table~\ref{tab:summary} presents a comprehensive comparison of speedup, accuracy, and memory trade-offs across all baseline methods. We consider \itlp{} as the reference baseline for accuracy since it optimally minimizes the underlying energy function.
In terms of accuracy, \stlp{} achieve optimal performance and \dblp{}, remains very close to optimal, achieving approximately $99\%$ accuracy on average. However, the approximate variant of \stlp{}($\gamma$) suffers noticeable accuracy degradation.
Machine learning–based methods also show relatively lower accuracy.
In terms of computational performance, the non–machine learning approaches achieve speedups of up to 102$\times$, demonstrating the effectiveness of our connected-component–based initialization and incremental update strategy.
Our approach also significantly outperforms machine learning–based methods.
From a memory perspective, \stlp{} does not exploit graph sparsity and is therefore limited to handling graphs of only up to 50K vertices with an average degree of 5. In contrast, \stlp{}($\gamma$) can process graphs with up to 7M vertices by using approximate matrix inverse techniques. For A2LP and CAGNN, increasing the graph size beyond 50K vertices leads to accuracy dropping close to $50\%$, indicating difficulty in learning even a binary classification task. On the other hand, both \itlp{} and our proposed \dblp{} scale efficiently, processing graphs with up to 50M vertices stored in CSR format with average degrees up to 7.

\begin{table}
\caption{Comparison summary table }
\label{tab:summary}
\centering
\begin{tabular}{c|cc|cc|cc}
\hline
Method 
& \multicolumn{2}{c|}{Accuracy}
& \multicolumn{2}{c|}{Speedup}
& \multicolumn{2}{c}{Max Graph Size} \\
& Avg & Max & Avg & Max & Node & Degree \\
\hline
\itlp{}            & 100 & 100 & 13   & 102  & 50M & 7 \\
\stlp{}            & 100  & 100 & 7    & 39   & 50K & 5  \\
\stlp{}($\gamma$)  & 70  & 83  & 7    & 7    & 7M  & 5  \\
CAGNN              & 76  & 88  & 32   & 32   & 5M & 5  \\
A2LP               & 56  & 58  & 1,935 & 1,935 & 50K & 5  \\
\dblp{}            & 99  & 100 & 1    & 1    & 50M & 7 \\
\hline
\end{tabular}

\end{table}


\section{Conclusion}

We propose \dblp{}, a scalable and efficient framework for label propagation that eliminates redundant computation in dynamic batch updates. 
\dblp{} is designed to exploit graph sparsity to improve both runtime and memory efficiency on large graphs.

We compare \dblp{} with multiple state-of-the-art baselines, covering classical optimization-based techniques and recent learning-based approaches, to quantify the trade-offs among accuracy, speed, and scalability. Across experiments, \dblp{} exhibits near-linear scaling with respect to the number of update batches. 
Our connected component–based initialization in \dblp{} substantially reduces the number of required iterations, yielding further computational savings and faster updates.
Overall, \dblp{} consistently provides better speed and scalability than competing methods while maintaining accuracy close to the optimal solution.

\dblp{} is currently designed for binary classification. In future work, we plan to extend the approach to the multi-class setting.

\bibliographystyle{ACM-Reference-Format}
\bibliography{sample-base}

\end{document}